\documentclass[letterpaper, 10 pt, conference]{ieeeconf}  

\IEEEoverridecommandlockouts                              

\usepackage{graphicx}
\usepackage{graphics} 
\usepackage{times} 
\usepackage{amsmath} 
\usepackage{amssymb} 
\usepackage[ruled, vlined, linesnumbered]{algorithm2e}
\usepackage{pseudocode}
\let\labelindent\relax
\usepackage{enumitem}
\usepackage{cite}
\usepackage{url}
\usepackage{listings}
\usepackage{color}
\usepackage{nicefrac}
\usepackage{pgfplots}
\pgfplotsset{compat=1.16}
\usepackage{standalone}

\newtheorem{theorem}{\bf Theorem}

\newtheorem{remark}{\bf Remark}

\newtheorem{proposition}{\bf Proposition}
\newtheorem{lemma}{\bf Lemma}
\newtheorem{assumption}{\bf Assumption}

\def\QED{~\rule[-1pt]{5pt}{5pt}\par\medskip}

\DeclareMathOperator*{\argmin}{arg\,min}

\title{\LARGE \bf
Point-Based Value Iteration and Approximately Optimal \\ Dynamic Sensor Selection for Linear-Gaussian Processes* 
}

\author{Michael Hibbard$^{1}$ \and Kirsten Tuggle$^{2}$ \and Takashi Tanaka$^{1}$
\thanks{*This work is supported by DARPA Grant D19AP00078,  FOA-AFRL-AFOSR-2019-0003 and NSF Award 1944318.}
\thanks{$^{1}$M. Hibbard and T. Tanaka are with the Department of Aerospace Engineering and Engineering Mechanics, University of Texas at Austin, TX 78712, USA.
        {\tt\small \{mwhibbard,ttanaka\}@utexas.edu}}%
\thanks{$^{2}$Kirsten Tuggle recently graduated from the Department of Aerospace Engineering and Engineering Mechanics, University of Texas at Austin, TX 78712, USA.,
        {\tt\small ktuggle@utexas.edu}}%
}

\begin{document}

\maketitle
\thispagestyle{empty}
\pagestyle{empty}

\begin{abstract}

The problem of synthesizing an optimal sensor selection policy is pertinent to a variety of engineering applications ranging from event detection to autonomous navigation. We consider such a synthesis problem over an infinite time horizon with a discounted cost criterion.
We formulate this problem in terms of a value iteration over the continuous space of covariance matrices. To obtain a computationally tractable solution, we subsequently formulate an approximate sensor selection problem, which is solvable through a point-based value iteration over a finite ``mesh" of covariance matrices with a user-defined bounded trace. We provide theoretical guarantees bounding the suboptimality of the sensor selection policies synthesized through this method and provide numerical examples comparing them to known results.
\end{abstract}


\section{Introduction}

We study the problem of strategic sensor selection for linear, time-invariant, Gaussian systems with a discounted cost over an infinite time horizon. The need to strategically allocate sensor resources naturally arises in numerous engineering applications, such as the situation where a certain sensor measurement is prohibitively expensive, or where the computational capability of an agent restricts the number of sensors that can be used simultaneously. 
Related to such situations is the field of active sensing, particularly in regard to the control of multiple sensing agents \cite{freundlich2015optimal},\cite{atanasov2013nonmyopic},\cite{kantaros2019asymptotically}, wherein the control of each agent must be determined in order to obtain the desired sensor measurements. 
%
In general, the question of optimal state estimation for linear systems subject to constraints on the total number of allowable sensor activations is well studied. The problem has been shown to be NP-hard for different mean square error (MSE)-based objective functions in finite \cite{williamson2011design} and infinite \cite{zhang2017sensor} time horizons. 

Many recent works have addressed this problem through formulations based on semidefinite programming relaxations.
In \cite{joshi2008sensor}, it was shown that the sensor selection problem for static systems (or dynamic systems without process noise) over finite time horizons is a mixed integer semidefinite programming problem (MISDP). The same classification was subsequently made for dynamic systems with process noise in \cite{mo2011sensor}. As a result, many methods exist that utilize the SDP nature of the relaxed problem, which is obtained by relaxing the Boolean selection variables to the continuum $[0,1]$. In \cite{joshi2008sensor}, rounding, together with a swapping and ordering scheme, was applied to select a subset of sensors satisfying the sensing constraint. In both \cite{mo2011sensor} and \cite{roy2016spatio}, the relaxed SDP was solved with an iterative, re-weighted $l$$-$$1$ constraint on the vector containing all scheduling variables to induce sparsity.


Recently, alternative approaches to the sensor selection problem have focused on leveraging the properties of submodular optimization to obtain computationally tractable solutions. Since it has been shown that the MSE is not submodular in general \cite{jawaid2015submodularity}, many works \cite{shamaiah2010greedy}, \cite{tzoumas2016sensor} have instead considered surrogate objectives with proven submodularity. However, performance guarantees are only valid with respect to the surrogate objective. Other approaches maintain the original objective but prove performance via characterizing weak submodularity \cite{hashemi2020randomized}, \cite{chamon2020approximate} for certain special cases of linear-Gaussian systems over finite time horizons. Similarly, a few works \cite{jawaid2015submodularity}, \cite{singh2017supermodular} provide special cases of linear-Gaussian systems for which the MSE is submodular, and provide sufficient conditions for which these cases can be established. Several other noteworthy methods include pruning-based approaches \cite{vitus2012efficient} as well as stochastic sensor selection \cite{gupta2006stochastic}.

%
We develop an alternative approach based on standard point-based value iteration (PBVI) for partially observable Markov decision processes (POMDPs) \cite{pineau2003point}. In PBVI, a representative set of belief points is selected, and a value iteration procedure is applied exclusively to these sample points. By considering only this representative sample, PBVI is able to avoid the curse of dimensionality typically associated with planning in POMDPs. 
Using this approach, we impose a cumulative objective that optimizes performance across both the transient and the steady-state phases of the system, whereas many existing works \cite{gupta2006stochastic}, \cite{zhang2017sensor} consider applications where the estimation performance is relevant only ``at infinity" and exclusively optimize the steady-state covariance.  %
%
Several existing works \cite{vaisenberg2014scheduling}, \cite{atia2011sensor}, \cite{ghasemi2018perception} formulate scheduling or activations as a POMDP and apply some version of PBVI; however, these works typically provide suboptimality bounds only for discrete state spaces.

Based on the standard PBVI procedure, we propose a novel approximate dynamic programming algorithm for the sensor selection problem wherein we represent the belief space as the space of positive semidefinite covariance matrices. To obtain our representative sample of belief points, we discretize the space of bounded covariance matrices through a uniform ``mesh-grid" procedure with a user-defined resolution. Considering the value iteration procedure exclusively over this finite set, we quantify the upper bound on the suboptimality of the resulting sensor selection policy in terms of the chosen resolution of the discretization procedure. We additionally show that this solution method is massively parallelizable, and provide numerical experiments showing that it performs favorably to existing methods. 

\subsection*{Notation}
We denote the set of all real numbers by $\mathbb{R}$, the set of all integers by $\mathbb{Z}$, and the set of all $n$$\times$$n$ integer-valued matrices by $\mathbb{Z}^{n\times n}$. Furthermore, we denote the set of all positive semidefinite matrices of size $n$$\times$$n$ by $\mathbb{S}_+^n$. For a matrix $A$$\in$$\mathbb{R}^{n \times n}$, we denote its trace by $\text{Tr}(A)$$=$$\sum_{i=1}^{n}a_{ii}$.

\section{Problem formulation}

We consider the discrete-time system with linear dynamics
\begin{equation}
\label{eq:state1}
x_{t+1}=Ax_t+w_t, \;\; w_t\sim\mathcal{N}(0,W), \;\; t=0, 1, \cdots,
\end{equation}
where $x_t$$\in$$\mathbb{R}^n$. We assume that the initial state distribution $x_0$$\sim$$\mathcal{N}(0, P_0)$ is known a priori and that the matrix $A$$\in$$\mathbb{R}^{n\times n}$ is Schur stable; i.e., each of its eigenvalues $\lambda_{i}$ satisfies $|\lambda_i(A)|$$<$$1$ for $i$$=$$ 1,\ldots,n$. The Schur stability of $A$ will be necessary for the existence of a sensor selection policy for our problem, 
as well as for bounding the suboptimality of this policy. We further assume that, at each time step, a set of $m$ sensors is available to observe the underlying state of the system. The linear sensor measurements follow 
\begin{equation}
\label{eq:state2}
y_t=Cx_t+v_t, \;\; v_t\sim\mathcal{N}(0,V), \;\; t=1, 2, \ldots,
\end{equation}
where the matrix $C$$\in$$\mathbb{R}^{n\times m}$. We interpret the $i$-th block row of \eqref{eq:state2} as the output of the $i$-th sensor. Due to the costs associated with operating every sensor at each time step, it may be advantageous to select a subset of these sensors to observe the underlying state. If a subset of sensors $S_t$$\subseteq$$\{1,\cdots, m\}$ is used at time step $t$$+$$1$, then 
\eqref{eq:state2} becomes
\begin{equation} \label{eq:state3}
y_{t+1}=C_{S_t}x_{t+1}+v_{t+1}, \; v_{t+1}\sim\mathcal{N}(0,V_{S_t}), \; t=0, 1, \cdots
\end{equation}
%
where $C_{S_t}$ and $V_{S_t}$ are submatrices of $C$ and $V$, respectively, formed by selecting the rows of $C$ and the rows and columns of $V$ corresponding to the index set $S_t$. Let $\{S_t\}_{t=0,1,\ldots}$ be the sequence of subsets of selected sensors. 
By the Kalman filter formula, the estimation error covariance
\begin{equation}
P_{t}=\mathbb{E}[(x_t-\hat{x}_t)(x_t-\hat{x}_t)^\top|y_1,\ldots,y_t],
\end{equation}
with $\hat{x}_t=\mathbb{E}[x_t|y_1, \cdots, y_t]$, is a function of the set $S_{t}$, and is recursively computed according to $P_{t+1}$$=$$f(P_t, S_t)$, where
\begin{equation}\label{eq:state_p}
f(P_t, S_t):=\left((AP_tA^\top+W)^{-1}+C_{S_t}^\top V_{S_t}^{-1}C_{S_t}\right)^{-1}
\end{equation}
for $t$$=$$0, 1, \ldots$. 
Now, let $\pi$$:$$\mathbb{S}_+^n$$\rightarrow$$\{0, 1, \ldots, m\}$, $S_t$$=$$\pi(P_t)$ be a time-invariant sensor selection policy, and let $\Pi$ denote the space of all such (Borel measurable) policies. We formulate our sensor selection problem as 
\begin{equation}\label{prob:1}
\min_{\pi\in\Pi} \quad \sum\nolimits_{t=0}^\infty \beta^t c(P_t, S_t),
\end{equation}
where $c(P_t, S_t)$ is a cost function and $0$$\leq$$\beta$$<$$1$ is a discount factor weighting the relative importance of present and future costs.
Our choice of incorporating a discount factor will be critical in concisely obtaining the suboptimality bounds for our method as alternative undiscounted, average-per-time-step cost characterizations often require further restrictions 
\cite{bertsekas2011dynamic}.
Due to the time-invariance of both the state dynamics in \eqref{eq:state_p} and the discounted cost, the assumption that the optimal policy is time-invariant is made without loss of generality (see, e.g., Section 3.7 of \cite{bertsekas2011dynamic}). We assume that $c(P_t, S_t)$ follows\footnote{More general cost functions $c(P_t, S_t)$$=$$Tr(\Phi P_t)$$+$$g(S_t)$, where $\Phi$ is a positive definite weight matrix, can be incorporated in the following analysis with minor modifications.} 
\[
c(P_t, S_t)=\text{Tr}(P_t) + g(S_t),
\]
where $\text{Tr}(P_t)$ measures the spread of the covariance matrix $P_t$, $g$$:$$2^{\{1, 2, \cdots, m\}}$$\rightarrow$$ [0,\infty)$ is a set function that maps subsets of active sensors to an associated cost.
Some set functions of interest include $g(S_t)$$=$$|S_t|$, which penalizes the number of sensors used simultaneously, and
\begin{equation}
\label{eq:exactly_one}
g(S_t)=\begin{cases}
0 & \text{ if $|S_t|=1$} \\
+\infty & \text{ otherwise} 
\end{cases}
\end{equation}
which stipulates the use of one sensor at each time step.

For computational tractability, we adopt a slight modification to the problem \eqref{prob:1}. First, define the set of all positive semidefinite matrices in $\mathbb{R}^{n\times n}$ with trace at most $\gamma$ as $\mathbb{S}_+^n(\gamma)$; i.e., $\mathbb{S}_+^n(\gamma)$$:=$$\{P$$\in$$\mathbb{S}_+^n$$:$$ \text{Tr}(P)$$\leq$$\gamma\}$. Now, let $\Pi(\gamma)$ be the set of stationary sensor selection policies $S_t$$=$$\pi(P_t)$ for which the set $\mathbb{S}_+^n(\gamma)$ is invariant; i.e., over repeated sensor measurement updates, the covariance matrix remains in the set $\mathbb{S}_+^n(\gamma)$.
\begin{assumption} \label{asmp1}
The set $\Pi(\gamma)$ of policies $\pi$ under which $\mathbb{S}_+^{n}(\gamma)$ is invariant is not empty.
\end{assumption}
In what follows, we study the modified problem:
\begin{equation} \label{prob:2}
\min_{\pi\in\Pi(\gamma)} \quad \sum\nolimits_{t=0}^\infty \beta^t c(P_t, S_t).
\end{equation}
As opposed to \eqref{prob:1}, $P_t$$\in$$\mathbb{S}_+^n(\gamma)$ is a hard constraint in \eqref{prob:2}. 
%
%
We adopt \eqref{prob:2} partly due to the fact that its bounded state space is more amenable for the analysis in Sections \ref{sec:exact_value_iteration}-\ref{sec:modified_val_iter_analysis}. Note that Assumption~\ref{asmp1} is not restrictive, since under the assumption that $A$ is Schur stable, there always exists a finite $\gamma$ such that Assumption~\ref{asmp1} holds.
We further assume that the initial covariance matrix $P_{0}$ satisfies $\text{Tr}(P_{0})$$\leq$$\gamma$. Again, this assumption is not restrictive. By the Schur stability assumption, there exists a finite $\gamma$ such that $\lim_{t\rightarrow \infty}\text{Tr}(P_{t})$$<$$ \gamma$, even under the trivial sensor selection policy $S_{t}$$=$$\{\emptyset\}$. Afterwards, the synthesized stationary policy can then be applied.

\section{Exact Value Iteration} \label{sec:exact_value_iteration}
In this section, we characterize the optimal solution to \eqref{prob:2} using exact value iteration.
To start with, define the space $B(\mathbb{S}_+^n(\gamma))$ of bounded functions on $\mathbb{S}_+^n(\gamma)$ by
\[
B(\mathbb{S}_+^n(\gamma)):=\{J:\mathbb{S}_+^n(\gamma)\rightarrow \mathbb{R}: \|J\|_\infty <\infty\}
\]
where $\|J\|_\infty$$:=$$\sup_{P\in \mathbb{S}_+^n(\gamma)} |J(P)|$ is the sup norm. Note that the space $B(\mathbb{S}_+^{n}(\gamma))$ equipped with the sup norm is complete \cite{completeness_lecture}.
Define the Bellman operator $T$ by
\begin{equation}
\label{eq:exact_vi}
(TJ)(P):=\min\nolimits_{S} \{c(P, S)+\beta J(f(P, S))\}.
\end{equation}

\begin{proposition} \label{prop1}
Under Assumption~\ref{asmp1}, the following hold:
\begin{itemize}
\item[(i)] $J$$\in$$B(\mathbb{S}_+^n(\gamma))$ implies that $TJ$$\in$$ B(\mathbb{S}_+^n(\gamma))$.
\item[(ii)] For all $J,J'$$\in$$B(\mathbb{S}_+^n(\gamma))$,
$\|TJ-TJ'\|_\infty$$\leq$$\beta \|J$$-$$J'\|_\infty$.
\item[(iii)] $\exists$ $J^*$$\in$$B(\mathbb{S}_+^n(\gamma))$ satisfying Bellman's equation
\begin{equation}
\label{eq:bellman}
J^*(P)=\min\nolimits_{S} \{ c(P, S)+\beta J^*(f(P, S))\}.
\end{equation}
\item[(iv)] For every $J$$\in$$B(\mathbb{S}_+^n(\gamma))$, the sequence $J_k$$=$$T^k J$ converges uniformly to $J^*$; i.e., $\lim_{k\rightarrow \infty} \|J_k-J^*\|_\infty$$=$$0$.
\end{itemize}
\end{proposition}
\begin{proof}
(i):  Since Assumption~\ref{asmp1} guarantees the existence of a sensor selection $S$ such that $f(P,S)$$\in$$\mathbb{S}_+^n(\gamma)$,
$
(TJ)(P)$$=$$\min_S \{ c(P, S)$$+$$\beta J(f(P,S))\}
$
is finite for $J$ bounded. Thus $J$$ \in $$B(\mathbb{S}_+^n(\gamma))$ implies $TJ$$ \in$$ B(\mathbb{S}_+^n(\gamma))$.

(ii): Let $q$$:=$$\|J-J'\|_\infty$$=$$\sup_{P\in\mathbb{S}_+^n(\gamma)}|J(P)$$-$$J'(P)|$. Then,
\begin{equation}
\label{eq:jm}
J(P)-q \leq J'(P) \leq J(P)+q
\end{equation}
for every $P$$\in$$\mathbb{S}_+^n(\gamma)$. Applying $T$ to each side of \eqref{eq:jm}, and noticing that
\begin{align*}
\min_S \{ & c(P, S)+\beta \left(J(f(P,S))\pm q\right)\} \\
&=\min_S \{ c(P, S)+\beta \left(J(f(P,S))\right)\}\pm \beta q,
\end{align*}
we have
\[
(TJ)(P)-\beta q \leq (TJ')(P) \leq (TJ)(P)+\beta q,
\]
for every $P$$\in$$\mathbb{S}_+^n(\gamma)$. Thus, $\|TJ-TJ'\|_\infty$$\leq$$\beta \|J-J'\|_\infty$.

(iii) and (iv): Since the space $B(\mathbb{S}_+^n(\gamma))$ equipped with the sup norm is a complete metric space, and the operator $T$ a contraction mapping, we can directly apply the Banach fixed point theorem \cite{khamsi2011introduction} to obtain the desired results.
\end{proof}

By Proposition~\ref{prop1}-(iv), one can obtain $J^*$ through a value iteration over  $\mathbb{S}_+^n(\gamma)$. 
Once $J^*$ is obtained, the optimal policy $\pi^*$$\in$$\Pi$ is then characterized according to
\[
\pi^*(P)=\argmin_{S} \{c(P, S)+\beta J^*(f(P, S))\}.
\]
Unfortunately, the value iteration procedure described in \eqref{eq:exact_vi} is not practical to implement as the function $J$ must be evaluated everywhere in the continuous space $\mathbb{S}_+^n(\gamma)$.
 
\section{Approximate value iteration} \label{sec:approx_value_iteration}
Due to the impracticality of continuous state-space value iteration, we now present a computationally tractable procedure wherein we approximate $J^{*}$
through a value iteration over a prespecified, finite set of sample covariance matrices.

\subsection{Mesh grid on $\mathbb{S}_+^n(\gamma)$}
Let $\epsilon$$>$$0$ be a fixed parameter that encodes the resolution of our mesh grid on $\mathbb{S}_{+}^{n}(\gamma)$, where a smaller $\epsilon$ yields a finer resolution. 
To formally define this mesh, consider the set
\begin{align*}
    \mathbb{M} & :=\{\epsilon P : P \in \mathbb{Z}^{n\times n}, P\succeq 0\}, \nonumber
\end{align*}
obtained by scaling the space of $n$$\times$$n$ positive-definite, symmetric, integer-valued matrices by our resolution parameter $\epsilon$. From this set $\mathbb{M}$, we subsequently define the subset 
\begin{align*}
    \mathbb{M}(\gamma) & :=\{\hat{P}\in\mathbb{M}: \text{Tr}(\hat{P})\leq \gamma \}, \nonumber
\end{align*}
where we have extracted the elements $\hat{P}$ from $\mathbb{M}$ with trace at most $\gamma$. Since we began with the discrete set of integer-valued matrices, the set $\mathbb{M}(\gamma)$$\subset$$\mathbb{S}_+^n(\gamma)$ consists of a finite number of elements and serves as a ``mesh grid'' on $\mathbb{S}_+^n(\gamma)$.
Table~\ref{table:1} shows the number of elements $|\mathbb{M}(\gamma)|$ contained in this set depending on the parameters of the system. 
These values were obtained through a brute-force enumeration of all matrices $\epsilon P$ satisfying $\text{Tr}(\epsilon P)$$\leq$$\gamma$ with $P$$\in$$\mathbb{Z}^{n\times n}$ and $P$$\succeq$$0$.\footnote{The authors are unaware of an efficient method to construct $\mathbb{M}(\gamma)$.}
As evident in Table \ref{table:1}, this uniform meshing procedure causes $|\mathbb{M}(\gamma)|$ to grow rapidly with $n$. We use this mesh to facilitate the analysis, however, future research should seek sampling methods that maintain the theoretical guarantees.
%
\begin{table}[t]
\centering
\begin{tabular}{|r|r|r|r|r|}
\hline
  & $\gamma=10$     & $\gamma=20$       &$\gamma=30$      & $\gamma=40$    \\ \hline
$n=2$ & $312$    & $2,261$      & $7,416$     & $17,349$ \\ \hline
$n=3$ & $9,888$   & $507,745$    & $5,487,604$  &     $30,105,633$  \\ \hline
$n=4$ & $217,905$ & $133,895,766$ &      --    &  --      \\ \hline
\end{tabular}
\caption{Number of elements in $\mathbb{M}(\gamma)$ when $\epsilon=1$.}
\label{table:1}
\end{table}
%
 
\subsection{Modified Bellman operator} %

We now introduce a modified value iteration that can be performed exclusively over the finite set $\mathbb{M}(\gamma)$. We will use this modified value iteration to approximate the exact value function $J$ in \eqref{eq:exact_vi}. Define the quantizer $\Theta: \mathbb{S}_+^n$$\rightarrow$$\mathbb{M}$ by
\begin{equation}\label{eq:theta1}
\Theta(P):=\epsilon \text{round}(\nicefrac{1}{\epsilon}P)+\epsilon nI,
\end{equation}
%
where the operator $\text{round}(\cdot)$ rounds each element of the input matrix to the nearest integer. Due to the effects of this rounding, $\epsilon \text{round}(\nicefrac{1}{\epsilon}P)$ in \eqref{eq:theta1} need not upper bound $P$ in the sense that $\epsilon \text{round}(\nicefrac{1}{\epsilon}P)$$\succeq$$P$. As we will formally prove in Lemma~\ref{lem:theta_omega}, the inclusion of the $\epsilon nI$ term in \eqref{eq:theta1} ensures that $\Theta(P)$$\succeq$$P$ does indeed hold. This inequality will be crucial for ensuring that the approximate value functions discussed in the sequel upper bound their corresponding exact values. 
\begin{remark}
Alternatively, adopting the quantizer
\begin{equation}
\label{eq:theta2}
\Theta'(P)=\argmin\nolimits_{Q\in\mathbb{M}}\{\text{Tr}(Q): Q \succeq P\}
\end{equation}
provides a tighter bound than \eqref{eq:theta1} in that $\Theta(P)$$\succeq $$\Theta'(P)$$\succeq$$P$.
However, \eqref{eq:theta2} is computationally difficult to implement. In the numerical examples in Section~\ref{sec:simulation}, we adopt
\begin{equation}
\Theta''(P)=\epsilon  \text{round}(\nicefrac{1}{\epsilon}P+t^*I)  \label{eq:theta3} 
\end{equation}
with $t^*$$=$$\min \{t$$\in$$\mathbb{R} : \epsilon  \text{round}(\nicefrac{1}{\epsilon}P+tI) \succeq P\}$,
which is easier to implement than \eqref{eq:theta2} yet still provides a tighter upper bound than \eqref{eq:theta1}. For the ease of analysis, we adopt \eqref{eq:theta1} to concisely analyze the gap between the corresponding approximate and exact value iterations, however, the theoretical results remain valid when \eqref{eq:theta2} and \eqref{eq:theta3} are incorporated.    \hfill \QED
\end{remark}
Now, introduce a modified Bellman operator $\bar{T}$ defined by
\begin{equation}
\label{eq:tbar}
(\bar{T}J)(P):=\min\nolimits_{S} \{c(P, S)+ \beta J(\Theta(f(P, S)))\}. 
\end{equation}

\subsection{Modified value iteration and suboptimal sensor selection}

For each $J$$\in$$B(\mathbb{S}_+^n(\gamma))$, the modified value iteration is defined by the sequence $\bar{J}_k$$=$$\bar{T}^kJ$.
In the following section, we formally prove that the sequence $\bar{J}_k$ converges uniformly to the unique solution $\bar{J}^*$$\in$$ B(\mathbb{S}_+^n(\gamma))$ of the Bellman equation
\begin{equation}
\label{eq:bellman_j_bar}
\bar{J}^*(P)=\min\nolimits_{S} \{ c(P, S)+\beta \bar{J}^*(\Theta(f(P, S)))\},
\end{equation}
under an appropriate assumption (Proposition~\ref{prop2}).

Notice that $\bar{J}_k$$=$$\bar{T}^kJ$ requires updates only on the finite set $\mathbb{M}(\gamma)$. 
Denote by $\bar{J}_k|_{\mathbb{M}(\gamma)}$$:$$\mathbb{M}(\gamma)$$\rightarrow$$ \mathbb{R}$ the restriction of $\bar{J}_{k}$ to $\mathbb{M}(\gamma)$.
Since each $\bar{J}_{k+1}$ is a function of $\bar{J}_k|_{\mathbb{M}(\gamma)}$ only, the modified value iteration can be performed exclusively on the finite set $\mathbb{M}(\gamma)$. Particularly, for each $P$$\in$$\mathbb{M}(\gamma)$ in each iteration, one can compute the value function update as
\begin{equation}
\label{eq:finite_vi}
\bar{J}_{k+1}|_{\mathbb{M}(\gamma)}(P)=\min\nolimits_S \{c(P,S)+\beta \bar{J}_k|_{\mathbb{M}(\gamma)}(\Theta(f(P,S)))\}
\end{equation}
Through this procedure, the convergence of the value functions to their optimal values, i.e., $\lim_{k\rightarrow \infty} \bar{J}_k|_{\mathbb{M}(\gamma)}$$=$$ \bar{J}^*|_{\mathbb{M}(\gamma)}$, follows directly from the uniform convergence of $\bar{J}_k$.
Once $\bar{J}^*|_{\mathbb{M}(\gamma)}$ is obtained, the solution to \eqref{eq:bellman_j_bar} can be recovered as
\[
\bar{J}^*(P)=\min\nolimits_S \{c(P,S)+\beta \bar{J}^*|_{\mathbb{M}(\gamma)}(\Theta(f(P,S)))\}.
\]
The suboptimal policy $\pi'$$:$$\mathbb{S}_+^n(\gamma)$$\rightarrow$$ 2^{\{1, 2, \cdots, m\}}$ corresponding to the Bellman equation \eqref{eq:bellman_j_bar} is likewise obtained  by
\begin{equation}
\label{eq:subopt_policy}
\pi'(P)=\argmin\nolimits_{S} \{c(P, S)+\beta \bar{J}^*(f(P,S))\}.
\end{equation}
Let $J^{\pi'}$$:$$\mathbb{S}_+^n(\gamma)$$\rightarrow$$\mathbb{R}$ be the value function corresponding to the policy $\pi'$ characterized as the unique solution
%
to 
\[
J^{\pi'}(P)=c(P, \pi'(P))+\beta J^{\pi'}(f(P, \pi'(P))).
\]
We will establish that $\|J^*$$-$$J^{\pi'}\|_\infty$$\leq$$\|J^*-\bar{J}^*\|_\infty$$\leq$$\frac{2\epsilon n^2}{(1-\beta)^2}$ under an assumption based on the main results (Theorems \ref{theo:main} and \ref{theo:policy}).
This bound implies that the performance gap between the optimal policy and the suboptimal policy obtained by the approximate value iteration $\bar{J}_k$$=$$\bar{T}^k J_0$ can be made arbitrarily small by decreasing $\epsilon$, although there is a corresponding increase in computational costs due to increasing $|\mathbb{M}(\gamma)|$.

\section{Analysis of modified value iteration} \label{sec:modified_val_iter_analysis}

Unfortunately, the performance of the modified value update in~\eqref{eq:tbar} is difficult to analyze. For this purpose, we introduce a third Bellman operator $\Omega$$:$$\mathbb{S}_+^n$$\rightarrow$$ \mathbb{S}_+^n$ defined by
\[
\Omega(P):=P+2\epsilon n I.
\]

The following are basic properties of $\Theta(\cdot)$ and $\Omega(\cdot)$:
\begin{lemma} \label{lem:theta_omega}
For each $P$$\succeq$$0$, we have $P$$\prec$$\Theta(P)$$\prec$$\Omega(P)$.
\end{lemma}
\begin{proof}
Define $R$$=$$\text{round}\left(\frac{1}{\epsilon}P\right)$$-$$\frac{1}{\epsilon}P$. By construction, $R$ is symmetric and each entry $R_{ij}$ satisfies $|R_{ij}|$$\leq$$ \frac{1}{2}$. Consequently, the eigenvalues of $R$ are bounded as
\[
\max_i |\lambda_i(R)|\leq \sqrt{\text{Tr}(R^\top R)}=\sqrt{\sum_{i,j} R_{ij}^2}\leq \sqrt{\frac{n^2}{4}}= \frac{n}{2}<n.
\]
Therefore,
\begin{align*}
\Theta(P)-P&=\epsilon \left(\text{round}\left(\frac{1}{\epsilon}P\right)-\frac{1}{\epsilon}P+nI\right) \\
&=\epsilon(R+nI)\succ 0. \\
\Omega(P)-\Theta(P)&=\epsilon \left(\frac{1}{\epsilon}P-\text{round}\left(\frac{1}{\epsilon}P\right)+nI\right) \\
&=\epsilon(-R+nI)\succ 0,
\end{align*}
where the positive definiteness of $(-R+nI)$ follows from our bound on the eigenvalues of $R$.
\end{proof}
Finally, introduce a new Bellman operator $\hat{T}$ defined by
\begin{equation}\label{eq:that}
(\hat{T}J)(P):=\min\nolimits_{S} \{c(P, S)+ \beta J(\Omega(f(P, S)))\}.
\end{equation}
for which we define $\hat{J}^*$$\in$$ B(\mathbb{S}_+^n(\gamma))$ according to
\begin{equation} \label{eq:bellman_j_hat}
\hat{J}^*(P)=\min\nolimits_{S} \{ c(P, S)+\beta \hat{J}^*(\Theta(f(P, S)))\},
\end{equation}

\subsection{Convergence of value iteration}

To proceed further, we make the following assumption.
\begin{assumption}\label{asmp2}
There exists $S$$\subseteq$$\{1, \cdots, N\}$ for each $P$$\in$$\mathbb{S}_+^n(\gamma)$ such that $\Omega(f(P,S))$$\in$$\mathbb{M}(\gamma)$.
\end{assumption}
\begin{remark}
Assumption~\ref{asmp2} is stronger than Assumption~\ref{asmp1}. However, while Assumption \ref{asmp2} is necessary for providing theoretical guarantees, numerical experiments indicate that the sensor selection policy synthesis proposed below often yields a favorable performance without this assumption. \hfill \QED
\end{remark}

We have the following results regarding the Bellman operators $\bar{T}$ and $\hat{T}$ as defined in \eqref{eq:tbar} and \eqref{eq:that}.
\begin{proposition} \label{prop2}
Under Assumption~\ref{asmp2}, the following hold:
\begin{itemize}
\item[(i)] $J$$\in$$B(\mathbb{S}_+^n(\gamma))$ implies that $\bar{T}J$ and $\hat{T}J$$\in$$ B(\mathbb{S}_+^n(\gamma))$.
\item[(ii)] For all $J,J'$$\in$$B(\mathbb{S}_+^n(\gamma))$,
$\|\bar{T}J-\bar{T}J'\|_\infty$$\leq$$\beta \|J-J'\|_\infty$ and $\|\hat{T}J-\hat{T}J'\|_\infty$$\leq$$\beta \|J-J'\|_\infty$.
\item[(iii)] $\exists$ $\bar{J}^*$, $\hat{J}^*$$\in$$B(\mathbb{S}_+^n(\gamma))$ satisfying Bellman's equation \eqref{eq:bellman_j_bar} and \eqref{eq:bellman_j_hat}, respectively.
\item[(iv)] For every $J$$\in$$B(\mathbb{S}_+^n(\gamma))$, define the sequences $\bar{J}_k$$=$$\bar{T}^k J$ and $\hat{J}_k$$=$$\hat{T}^k J$. Then, we have $\lim_{k\rightarrow \infty} \|\bar{J}_k-\bar{J}^*\|_\infty$$=$$0$ and $\lim_{k\rightarrow \infty} \|\hat{J}_k-\hat{J}^*\|_\infty$$=$$0$
\end{itemize}
\end{proposition}
\begin{proof}
Here we prove the aforementioned properties only for the operator $\hat{T}$, as those for $\bar{T}$ follow similarly.

(i):  Since Assumption~\ref{asmp2} guarantees the existence of a sensor selection $S$ such that $\Omega(f(P,S))$$\in$$\mathbb{M}(\gamma)$$\subset$$\mathbb{S}_+^n(\gamma)$,
$
(\hat{T}J)(P)$$=$$\min_S \{ c(P, S)$$+$$\beta J(\Omega(f(P,S)))\}
$
is finite for $J$ bounded. Thus $J$$ \in $$B(\mathbb{S}_+^n(\gamma))$ implies $\hat{T}J$$ \in$$ B(\mathbb{S}_+^n(\gamma))$.

(ii): Let $q$$:=$$\|J-J'\|_\infty$$=$$\sup_{P\in\mathbb{S}_+^n(\gamma)}|J(P)$$-$$J'(P)|$. Then,
\begin{equation}
\label{eq:jm}
J(P)-q \leq J'(P) \leq J(P)+q
\end{equation}
for every $P$$\in$$\mathbb{S}_+^n(\gamma)$. Applying $\hat{T}$ to each side of \eqref{eq:jm}, and noticing that
\begin{align*}
\min_S \{ & c(P, S)+\beta (J(\Omega(f(P,S)))\pm q)\} \\
&=\min_S \{ c(P, S)+\beta J(\Omega(f(P,S)))\}\pm \beta q,
\end{align*}
we have
\[
(\hat{T}J)(P)-\beta q \leq (\hat{T}J')(P) \leq (\hat{T}J)(P)+\beta q,
\]
for every $P$$\in$$\mathbb{S}_+^n(\gamma)$. Thus, $\|\hat{T}J-\hat{T}J'\|_\infty$$\leq$$\beta \|J-J'\|_\infty$.

(iii) and (iv): Since the space $B(\mathbb{S}_+^n(\gamma))$ equipped with the sup norm is a complete metric space, and the operator $\hat{T}$ a contraction mapping, we can again apply the Banach fixed point theorem to obtain the desired results.
\end{proof}
%

\subsection{Comparison of value iteration sequences}
Let $J$$\in$$B(\mathbb{S}_+^{n}(\gamma))$ be an arbitrary constant function.
The following are several basic properties of the three sequences $J_k$$=$$T^k J$, $\bar{J}_k$$=$$\bar{T}^k J$, and $\hat{J}_k$$=$$\hat{T}^kJ$.
\begin{lemma}\label{lem:monotone}
For each $k$$=$$0,1,2,...$ and for every $0$$\preceq$$P$$\preceq$$Q$, we have $J_k(P)$$\leq$$J_k(Q)$ and $\hat{J}_k(P)$$\leq$$\hat{J}_k(Q)$.
\end{lemma}
\begin{proof}
Here we prove only that $\hat{J}_k(P)$$\leq$$\hat{J}_k(Q)$, as the remaining inequality follows by a nearly identical argument. We prove the claim by induction. The claim trivially holds for $k$$=$$0$ since $\hat{J}_0$ is an arbitrary constant function by construction.
Suppose that $\hat{J}_k(P)$$\leq$$\hat{J}_k(Q)$ holds for some $k$$\geq$$0$. Now, for every $0$$\preceq$$ P'$$\preceq$$Q'$ and for every $S$, we have that  $c(P',S)$$\leq$$c(Q',S)$ and $\Omega(f(P',S))$$\leq$$ \Omega(f(Q',S))$. Thus,
\begin{align*}
\hat{J}_{k+1}(P')&=\min_S \{c(P', S)+\beta \hat{J}_k(\Omega(f(P',S)))\} \\
&\leq \min_S \{c(Q', S)+\beta \hat{J}_k(\Omega(f(Q',S)))\} \\
&=\hat{J}_{k+1}(Q'),
\end{align*}
completing the proof.
\end{proof}
\begin{proposition} \label{prop:j_bar_hat}
For each $k$$=$$0, 1, 2, \ldots$ and $P$$\succeq$$0$, we have $J_k(P)$$\leq$$\bar{J}_k(P)$$\leq $$\hat{J}_k(P)$.
\end{proposition}
\begin{proof}
We again prove the claim by induction. The claim trivially holds for $k$$=$$0$ as each $J$$\in$$B(\mathbb{S}_+^{n}(\gamma))$ is an arbitrary constant function. Assume that the claim holds for $k$$\geq$$0$. Then, for each $P'$$\succeq$$0$,
\begin{subequations}
\begin{align}
J_k(P')&\leq J_k(\Theta(P')) \label{prop4-1} \\
&\leq   \bar{J}_k(\Theta(P')) \label{prop4-2} \\
&\leq  \hat{J}_k(\Theta(P')) \label{prop4-3} \\
&\leq \hat{J}_k(\Omega(P')), \label{prop4-4}
\end{align}
\end{subequations}
where we first obtain the inequality in (\ref{prop4-1}) by recalling that $P'$$\preceq$$\Theta(P')$ by Lemma 1, which then implies that $J(P')$$\leq$$J(\Theta(P'))$ by Lemma 2. The two following inequalities in \eqref{prop4-2} and \eqref{prop4-3} are each obtained as a result of our induction hypothesis. Finally, we can again invoke the results of Lemmas 1 and 2 to obtain (\ref{prop4-4}) from \eqref{prop4-3}. Now, setting $P'$$=$$f(P,S)$, we have
\[
J_k(f(P,S)) \leq \bar{J}_k(\Theta(f(P,S))) \leq \hat{J}_k(\Omega(f(P,S)))
\]
for each $P$$\succeq$$0$ and $S$$\subseteq$$\{1,\ldots,N\}$. Thus,
\begin{align*}
\min_S & \{c(P,S)+\beta J_k(f(P,S))\} \\
&\leq  \min_S \{c(P,S)+\beta \bar{J}_k(\Theta(f(P,S)))\} \\
& \leq  \min_S \{c(P,S)+\beta \hat{J}_k(\Omega(f(P,S)))\}
\end{align*}
which implies that $J_{k+1}(P)$$\leq$$\bar{J}_{k+1}(P)$$\leq$$\hat{J}_{k+1}(P)$.
\end{proof}

\subsection{Upper bound of $\hat{J}_k$}

In Lemma~\ref{lem:monotone}, we have shown that $\hat{J}_k$ is monotonically non-decreasing. Proposition~\ref{prop:ub} below shows that the ``slope'' of  $\hat{J}_k$ is bounded. 
Define a convergent sequence by
\[
L_{k+1}=1+\beta L_k,\;\; L_0=0,
\]
where $\beta$ is our discount factor. By inspection, $L_{1},\ldots,L_{k},\ldots$ is the sequence of partial sums of a geometric series. It is straightforward to bound these terms by $L_k$$<$$ L$$:=$$\nicefrac{1}{1-\beta}$.
\begin{proposition} \label{prop:ub}
$\hat{J}_k(P$$+$$\Delta P)$$\leq$$\hat{J}_k(P)$$+$$L_k \text{Tr}(\Delta P)$ for every $k$$=$$0, 1, 2, \ldots$ and any $P,\Delta P$$ \succeq$$ 0$.
\end{proposition}
\begin{proof}
We prove the claim by induction. The claim trivially holds for $k$$=$$0$. Suppose that the claim holds for some $k$$\geq$$0$. For $P'$$\succeq$$0$ and $\Delta P'$$\succeq$$0$, by Lemma~\ref{lem:f_ub} (provided in the Appendix), we have that
\begin{equation}\label{eq:lemma4-apply}
    f(P'+\Delta P', S) \preceq f(P', S)+\Delta f
\end{equation}
for some $\Delta f$$ \succeq$$ 0$ chosen such that $\text{Tr}(\Delta f)$$\leq$$ \text{Tr}(\Delta P')$. Recalling the monotonicity of $\hat{J}_k$ by Lemma \ref{lem:monotone}, we have 
\begin{align*}
    \hat{J}_k(\Omega(f(P'+\Delta P', S)))&=\hat{J}_k(f(P'+\Delta P', S)+2\epsilon n I) \\
    &\leq \hat{J}_k(f(P', S) + \Delta f +2\epsilon n I ) \\
    &=\hat{J}_k(\Omega(f(P', S))+\Delta f),
\end{align*}
where the first equality is by the definition of $\Omega(\cdot)$, the subsequent inequality is obtained by substituting for the relation in \eqref{eq:lemma4-apply}, and the final equality is obtained by repackaging the first and third term in $\hat{J}_{k}(\cdot)$ into the definition of $\Omega(\cdot)$. Using this result,
\begin{subequations}
\begin{align}
& \hat{J}_{k+1}(P'+\Delta P') \nonumber \\
&=\min_S \{c(P',S)+\text{Tr}(\Delta P') \nonumber \\
& \qquad \qquad \qquad +\beta \hat{J}_k(\Omega(f(P'+\Delta P', S)))\} \label{eq:prop5-1} \\
&\leq \min_S \{c(P',S)+\text{Tr}(\Delta P') \nonumber \\
& \qquad \qquad \qquad +\beta \hat{J}_k(\Omega(f(P', S))+\Delta f)\} \label{eq:prop5-2} \\
&\leq \min_S \{c(P',S)+\text{Tr}(\Delta P') \nonumber \\ 
& \qquad \qquad \qquad +\beta \hat{J}_k(\Omega(f(P', S)))+\beta L_k \text{Tr}(\Delta f)\} \label{eq:prop5-3} \\
&\leq \min_S \{c(P',S) +\beta \hat{J}_k(\Omega(f(P', S))) \nonumber \\
& \qquad \qquad \qquad +(1+\beta L_k) \text{Tr}(\Delta P')\} \label{eq:prop5-4} \\
&=\min_S \{c(P',S) +\beta \hat{J}_k(\Omega(f(P', S)))\} \nonumber \\
& \qquad \qquad \qquad +L_{k+1}\text{Tr}(\Delta P') \label{eq:prop5-5} \\
&=\hat{J}_{k+1}(P')+L_{k+1}\text{Tr}(\Delta P') \label{eq:prop5-6},
\end{align}
\end{subequations}
where we have obtained the first equality by noting that $\text{Tr}(P'$$+$$\Delta P')$$=$$\text{Tr}(P')$$+$$\text{Tr}(\Delta P')$ and extracting the latter term from $c(P'$$+$$ \Delta P',S)$. We then obtain \eqref{eq:prop5-2} by applying our previous result. The inequality in \eqref{eq:prop5-3} then follows from the induction hypothesis. Recalling that $\Delta f$ is chosen such that $\text{Tr}(\Delta f)$$\leq$$\text{Tr}(\Delta P)$ holds by construction, we substitute for this relation and group the $\text{Tr}(\Delta P)$ terms to obtain \eqref{eq:prop5-4}. Noting the definition of $L_{k+1}$ and that $\text{Tr}(\Delta P)$ is independent of our sensor selection $S$ yields \eqref{eq:prop5-5}. Substituting for the definition of $\hat{J}_{k}(\cdot)$ then yields \eqref{eq:prop5-6}.
\end{proof}
By Proposition~\ref{prop:j_bar_hat}, we have $J_k$$\leq$$\bar{J}_k$$\leq \hat{J}_k$.
The next two propositions upper bound of the gap between $J_k$ and $\hat{J}_k$.
\begin{proposition} \label{prop:l_ub}
Suppose $\hat{J}_k$$=$$\hat{T}^k J_0$ where $J_0$$\in$$ B(\mathbb{S}_+^{n}(\gamma))$ is a constant function. For each $k$$=$$0,1,2, \ldots$, we have
\[
0\leq (\hat{T}\hat{J}_k)(P)-(T\hat{J}_k)(P)\leq 2\epsilon n^2 L \;\;\; \forall P \in \mathbb{S}_+^{n}(\gamma).
\]
\end{proposition}
\begin{proof}
The first inequality follows from the monotonicity of $\hat{J}_k$, as proven in Lemmas \ref{lem:theta_omega} and \ref{lem:monotone}:
\begin{align*}
(T\hat{J}_k)(P)&=\min_S \{c(P,S)+\beta \hat{J}_{k-1}(f(P,S))\} \\
&\leq \min_S \{c(P,S)+\beta \hat{J}_{k-1}(\Omega(f(P,S)))\} \\
&=(\hat{T}\hat{J}_k)(P).
\end{align*}
To see the second inequality,
\begin{subequations}
\begin{align}
(\hat{T} & \hat{J}_k)(P) = \min_S \{c(P,S)+\beta \hat{J}_{k-1}(\Omega(f(P,S)))\} \label{prop6-1} \\
&= \min_S \{c(P,S)+\beta \hat{J}_{k-1}(f(P,S)+2\epsilon n I)\} \label{prop6-2} \\
&\leq \min_S \{c(P,S)+\beta \hat{J}_{k-1}(f(P,S)) \nonumber \\
& \qquad \qquad \qquad \qquad \qquad +L_{k-1}\text{Tr}(2\epsilon n I)\} \label{prop6-3} \\
&\leq \min_S \{c(P,S)+\beta \hat{J}_{k-1}(f(P,S))\}+2\epsilon n^2 L \label{prop6-4} \\
&=(T\hat{J}_k)(P)+2\epsilon n^2 L \label{prop6-5},
\end{align}
\end{subequations}
where \eqref{prop6-2} follows \eqref{prop6-1} by the definition of $\Omega$. Then, \eqref{prop6-3} follows from the result of Proposition~\ref{prop:ub}. We subsequently obtain \eqref{prop6-4} by noting that the final term in \eqref{prop6-3} is independent of our sensor selection $S$ and recalling the fact that, for all $k$$=$$0,1,2,\ldots$, $L_{k}$$\leq$$L$. Finally, we obtain \eqref{prop6-5} by the definition of $(T\hat{J}_k)(P)$.
\end{proof}
\begin{proposition} \label{prop:j_hat}
Suppose $J_k$$=$$T^k J_0$ and $\hat{J}_k$$=$$\hat{T}^k J_0$ where $J_0$$\in$$ B(\mathbb{S}_+^{n}(\gamma))$ is a constant function. We have
\begin{equation*}
    \limsup_{k\rightarrow \infty} \|J_k-\hat{J}_k\|_\infty \leq \frac{2\epsilon n^2}{(1-\beta)^2}.
\end{equation*}
\end{proposition}

\begin{proof}
Notice that 
\begin{align}
\|J_{k+1}-\hat{J}_{k+1}\|_\infty &= \|TJ_k-\hat{T}\hat{J}_k\|_\infty \nonumber \\
&= \|TJ_k-T\hat{J}_k+T\hat{J}_k-\hat{T}\hat{J}_k\|_\infty \nonumber  \\
&\leq \|TJ_k-T\hat{J}_k\|_\infty + \|T\hat{J}_k-\hat{T}\hat{J}_k\|_\infty \nonumber  \\
&\leq \beta  \|J_k-\hat{J}_k\|_\infty + 2 \epsilon n^2 L, \label{eq:6_1}
\end{align}
where the final inequality is obtained using the fact that $T$ is a contraction mapping, as proven in Proposition~\ref{prop1}, to upper bound the first term, and that we can upper bound the second term by the result of Proposition~\ref{prop:l_ub}. Define a sequence $\ell_k$ by
\begin{equation} \label{eq:6_2}
    \ell_{k+1}=\beta \ell_k + 2\epsilon n^2 L
\end{equation}
with $\ell_0$$=$$0$. We have $\ell_k$$\leq$$\frac{2\epsilon n^2 L}{1-\beta}$$=$$\frac{2\epsilon n^2}{(1-\beta)^2}$ for $k$$=$$0,1,2,\ldots$. Comparing \eqref{eq:6_1} and \eqref{eq:6_2}, we have $\|J_k -\hat{J}_k\|_\infty$$\leq$$\ell_k$. Thus,
\[
\limsup_{k\rightarrow \infty} \|J_k-\hat{J}_k\|_\infty \leq \frac{2\epsilon n^2}{(1-\beta)^2},
\]
where we have substituted for $L$$=$$\frac{1}{1-\beta}$.
\end{proof}

\subsection{Main results}
Following from this analysis, we now present our main results, which are summarized in the following two theorems. 
\begin{theorem}\label{theo:main}
Let $J^*$ be the optimal value function characterized by \eqref{eq:bellman}. Let $\bar{J}_k$$=$$\bar{T}^k J_0$ be the sequence generated by the approximate value iteration with \eqref{eq:tbar} where $J_0$$\in$$B(\mathbb{S}_+^n(\gamma))$ is a constant function. We have
\[
\limsup_{k\rightarrow \infty} \|J^* -\bar{J}_k\|_\infty \leq \nicefrac{2\epsilon n^2}{(1-\beta)^2}.
\]
\end{theorem}
\begin{proof}
Since $0\leq$$J_k$$\leq$$\bar{J_k}$$\leq$$\hat{J}_k$, it follows from Proposition~\ref{prop:j_hat} that
$\limsup_{k\rightarrow \infty} \|J_k -\bar{J}_k\|_\infty$$\leq$$\frac{2\epsilon n^2}{(1-\beta)^2}$. Now,
\begin{align*}
& \limsup_{k\rightarrow \infty} \|J^*-\bar{J}_k\|_\infty
\leq \limsup_{k\rightarrow \infty} \|J^*-J_k+J_k-\bar{J}_k\|_\infty \\
&\qquad \quad \leq \underbrace{\limsup_{k\rightarrow \infty} \|J^*-J_k\|}_{=0 \text{ by Proposition~\ref{prop1}}}+\limsup_{k\rightarrow \infty} \|J_k-\bar{J}_k\|_\infty \\
&\qquad \quad =\limsup_{k\rightarrow \infty} \|J_k-\bar{J}_k\|_\infty \leq \nicefrac{2\epsilon n^2}{(1-\beta)^2}.
\end{align*}
\end{proof}

\begin{theorem} \label{theo:policy}
For all $P$$\succeq$$0$, $J^*(P)$$\leq$$J^{\pi'}(P)$$\leq$$\bar{J}^*(P)$.
\end{theorem}
\begin{proof}
Again, set $J_k$$=$$T^k J_0$ and $\bar{J}_k$$=$$\bar{T}^k J_0$ for $J_0\in B(\mathbb{S}_+^{n}(\gamma))$ a constant function. Define the sequence $J_k^{\pi'}$ by
\[
J_{k+1}^{\pi'}(P)=c(P, \pi'(P))+\beta J_k^{\pi'}(f(P, \pi'(P))).
\]
By construction, we have that 
\begin{equation}\label{eq:j_pi}
    J_k(P) \leq J_k^{\pi'}(P) \leq \bar{J}_k(P) \;\; \forall P \in \mathbb{S}_+^{n}(\gamma)
\end{equation}
for $k$$=$$0$. Now, Suppose \eqref{eq:j_pi} holds for some $k$$\geq$$0$. Then,
\begin{align*}
J_{k+1}(P)&=\min\nolimits_S \{c(P,S)+\beta  J_k(f(P,S))\} \\
&\leq c(P, \pi'(P))+\beta J_k(f(P, \pi'(P))) \\
&\leq c(P, \pi'(P))+\beta J_k^{\pi'}(f(P, \pi'(P))) =J_{k+1}^{\pi'}(P),
%
\end{align*}
where the second inequality follows from the induction hypothesis. Likewise,
\begin{align*}
J_{k+1}^{\pi'}(P)&= c(P, \pi'(P))+\beta J_k^{\pi'}(f(P, \pi'(P))) \\
&\leq c(P, \pi'(P))+\beta \bar{J}_k(f(P, \pi'(P))) \\
&=\min\nolimits_S \{c(P, S)+\beta \bar{J}_k(f(P, S)) \} =\bar{J}_{k+1}(P).
%
\end{align*}
Therefore, \eqref{eq:j_pi} holds for every $k$$=$$0, 1, \ldots$. 
Since $J_k$, $J_k^{\pi'}$, and $\bar{J}_k$ converge uniformly to $J^*$, $J^{\pi'}$ and $\bar{J}^*$, respectively, we have $J^*$$\leq$$J^{\pi'}$$\leq$$\bar{J}^*$. 
\end{proof}

Theorems \ref{theo:main} and \ref{theo:policy} imply that $\|J^*$$-$$J^{\pi'}\|_\infty$$\leq$$\frac{2\epsilon n^2}{(1-\beta)^2}$. Although the policy $\pi'$ is suboptimal, the performance loss can be made arbitrarily small by selecting a sufficiently small $\epsilon$.

\section{Numerical experiments}
\label{sec:simulation}
In this section, we use the approximate value iteration
\begin{equation} \label{eq:sim_vi}
\bar{J}_{k+1}(P)=(\bar{T}\bar{J}_k)(P) \quad \forall P \in \mathbb{M}(\gamma), \quad k=0, 1, \ldots
\end{equation}
where $\bar{T}$ is the Bellman operator defined in \eqref{eq:tbar} equipped with the quantizer $\Theta''(\cdot)$ defined in \eqref{eq:theta3}. Recalling that $\Theta''(\cdot)$ provides a tighter bound than $\Theta(\cdot)$, our theoretical bounds remain valid. 
Furthermore, note that the computation of the right hand side of \eqref{eq:sim_vi} for each $P$$\in$$\mathbb{M}(\gamma)$ is a function the previous value function iterates and $P$ itself. Since each value function update is independent of all other $P'$$ $$\in$$\mathbb{M}(\gamma)$, their computations are massively parallizable. We exploit this property by implementing the value function updates in \eqref{eq:sim_vi} on a GPU using the CUDA-enabled MATLAB interface. 
Through this implementation, each complete iteration of \eqref{eq:sim_vi} requires less than a second of computation time, even for meshes with cardinalities $|\mathbb{M}(\gamma)|$ of several million elements.

For our numerical analysis, we revisit the sensor selection problem of \cite{vitus2012efficient}, wherein a 3-dimensional process with four available sensors is considered. The system parameters are


\begin{align*}
A&\!=\!\begin{bmatrix}
-0.6 & 0.8 & 0.5 \\
-0.1 & 1.5 & -1.1 \\
1.1 & 0.4 & -0.2
\end{bmatrix}, \; 
C\!=\!\begin{bmatrix}
0.75 & -0.2 & -0.65 \\
0.35 & 0.85 & 0.35 \\
0.2 & -0.65 & 1.25 \\
0.7 & 0.5 & 0.5
\end{bmatrix},
\end{align*}
$W$$=$$I_{3\times3}$, and $V$$=$$\text{diag}(0.53,0.8,0.2,0.5)$, where $\text{diag}(\cdot)$ is a diagonal matrix whose entries are the listed values.
To penalize the choice of sensors, we adopt the set function $g(S_{t})$ defined in \eqref{eq:exactly_one}, which permits the selection of exactly one active sensor at each time step.
We additionally use a discount factor of $\beta$$=$$0.95$ and a trace bound of $\gamma$$=$$15$.

To study the impact of the resolution $\epsilon$ on the resulting value functions, we construct the set $\mathbb{M}(\gamma)$ for $\epsilon$ ranging from $\epsilon$$=$$3/7(=$$0.428)$ to $\epsilon$$=$$1$. For each $\mathbb{M}(\gamma)$, we perform the value iteration procedure in \eqref{eq:sim_vi} until convergence, which we observed to require no more than $500$ iterations. Once each value of $\bar{J}^*$ was obtained, we then extracted the corresponding sensor selection policy $\pi'$ through \eqref{eq:subopt_policy}. 

For each $\epsilon$, Fig.~\ref{fig:value_func} plots the value function $\bar{J}^*(P(\lambda))$ evaluated at $P(\lambda)$$=$$0.01\lambda I_3$, for $\lambda$$=$$1,\ldots,250$.
We additionally plot the estimated values of $J^{\pi'}(P(\lambda))$, which we obtain by simulating the trajectories $P_t$, $t$$=$$0,1\ldots$, starting with $P_0$$=$$P(\lambda)$ and repeatedly iterating \eqref{eq:state_p} using $C_{S_{t}}$ and $V_{S_{t}}$ corresponding to $\pi'(P_{t})$. Fig.~\ref{fig:value_func} confirms that $J^{\pi'}(P)$$\leq$$ \bar{J}^*(P)$ holds (Theorem~\ref{theo:policy}) for each value of $\epsilon$ considered. Although the exact value function $J^*(P)$ is not computable, Fig.~\ref{fig:value_func} indicates that a smaller $\epsilon$ tends to provide a tighter upper bound, given by $\bar{J}^*(P)$.
This relationship holds even if the gap
obtained in Theorem~\ref{theo:main} is overly conservative (in the case of $\epsilon$$=$$0.5$, for example, we have that $2\epsilon n^2/(1$$-$$\beta)^2$$=$$3600$).

\begin{figure}[t]
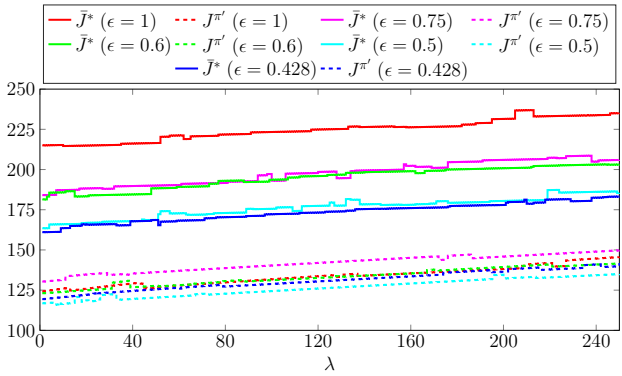

    \centering
    \includestandalone[width=0.95\columnwidth]{cost_func_comparisons}
    \caption{Value functions $\bar{J}^*(P)$ and $J^{\pi'}(P)$ for various values of $\epsilon$.
    }
    \label{fig:value_func}
\end{figure}

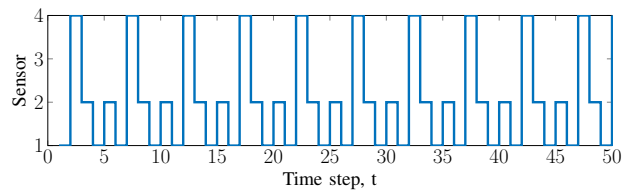
\begin{figure}[t]
    \centering
%
%
\definecolor{mycolor1}{rgb}{0.00000,0.44700,0.74100}%
\begin{tikzpicture}

\begin{axis}[%
width=6.5in,
height=1.5in,
at={(0,0)},
tick label style={font=\LARGE},
scale only axis,
xmin=0,
xmax=50,
xlabel style={font=\color{white!15!black},font=\LARGE},
xlabel={Time step, t},
ymin=1,
ymax=4,
ylabel style={font=\color{white!15!black},font=\LARGE},
ylabel={Sensor},
axis background/.style={fill=white}
]
\addplot[const plot, color=mycolor1, forget plot, line width=2] table[row sep=crcr] {%
1	1\\
2	4\\
3	2\\
4	1\\
5	2\\
6	1\\
7	4\\
8	2\\
9	1\\
10	2\\
11	1\\
12	4\\
13	2\\
14	1\\
15	2\\
16	1\\
17	4\\
18	2\\
19	1\\
20	2\\
21	1\\
22	4\\
23	2\\
24	1\\
25	2\\
26	1\\
27	4\\
28	2\\
29	1\\
30	2\\
31	1\\
32	4\\
33	2\\
34	1\\
35	2\\
36	1\\
37	4\\
38	2\\
39	1\\
40	2\\
41	1\\
42	4\\
43	2\\
44	1\\
45	2\\
46	1\\
47	4\\
48	2\\
49	1\\
50 4\\
};
\end{axis}

\end{tikzpicture}%
    \caption{Sequence of selected sensors when the initial state is $P_{0}=I_{3\times3}$. Sensor selection policy is computed with $\epsilon=0.428$.}
    \label{fig:periodic}
\end{figure}

We observe that the sequence of selected sensors under the synthesized sensor selection policy always eventually exhibits a periodic behavior. Using $\epsilon$$=$$0.428$, for example, the sequence of selected sensors eventually becomes a repetition of $\{4,2,1,2,1\}$ regardless of the initial covariance, as shown in Fig.~\ref{fig:periodic} for $P_{0}$$=$$I_{3\times3}$. 
We additionally observe that varying $\epsilon$ yields different periodic sequences (Table~\ref{table:2}). 
%
%
Interestingly, the sequences we obtain differ from the repeating sequence $\{4,1,4,2,1,2,3\}$ obtained in \cite{vitus2012efficient}. Although \cite{vitus2012efficient} considered an undiscounted problem, where the performance is evaluated by $\lim_{T\rightarrow \infty}\frac{1}{T}\sum_{t=1}^T \text{Tr}(P_t)$, it is noteworthy that some of the sequences we obtain outperform that of \cite{vitus2012efficient}, even under an undiscounted setting, as shown in Table~\ref{table:2}. This cost was also assessed for a greedy approach, i.e., beginning at the initial time and selecting the sensor at each time-step that minimizes the cost to that point. This steady-state sequence matches that of the proposed for $\epsilon$$=$$0.6,1$ but differs in the transient phase. Unlike greedy, our solution quality can be adapted via modulating parameter $\epsilon$. Finally, the cost from exclusively using the sensor minimizing the steady-state covariance \cite{zhang2017sensor} was also included, which demonstrates the potential unreliability in applying steady-state-only solutions to our particular problem.
Our simulations indicate that the sensor selection policy obtained by the proposed method performs well in practice, although the suboptimality guarantee provided by Theorem~\ref{theo:main} is conservative.

\begin{table}[t]
\centering
\begin{tabular}{|l|l|l|}
\hline
Method  & Sequence     & Undiscounted cost   \\ \hline
$\epsilon=0.5$ & $\{4,2,1\}$    & $6.4237$ \\ \hline
$\epsilon=0.428$ & $\{4,2,1,2,1\}$    & $6.6944$ \\ \hline
$\epsilon=0.6, \epsilon=1$ & $\{2,2,1\}$    & $6.8380$ \\ \hline
$\epsilon=0.75$& $\{2,2,2,1\}$    & $7.3535$ \\ \hline \hline
\cite{vitus2012efficient} & $\{4,1,4,2,1,2,3\}$    & $6.9410$ \\ \hline
\hline
greedy solution & \{2, 2, 1\} & 6.8380 \\ \hline \hline
\cite{zhang2017sensor} steady-state solution & \{2\} & 25.5572 \\ \hline
\end{tabular}
\caption{Repeating sequence and corresponding undiscounted cost.}
\label{table:2}
\end{table}

\section{Conclusion and Future Work}

We considered the problem of synthesizing an optimal sensor selection policy for linear-Gaussian systems, expressing it in terms of a value iteration over the continuous space of covariance matrices. We formulated a point-based value iteration approach over a prespecified, finite set of covariance matrices to obtain a computationally tractable procedure. This method is accompanied by proven bounds for suboptimality of the obtained sensor selection policy in terms of a configurable solution parameter. Favorable performance was demonstrated with respect to existing solutions.

Future research should seek to develop methods that efficiently construct $\mathbb{M}(\gamma)$, as na\"ively using a ``mesh grid" is impractical for systems with high-dimensional state spaces.
Similarly, as discussed in \cite{vitus2012efficient}, the sensor selection problem can be viewed as the dual to the optimal control problem in switched systems. Investigating how existing pruning methods in this context, such as that of \cite{andrien2020near}, can be applied to our sensor selection problem of interest, is a second straightforward research direction.
Finally, future work should additionally seek to obtain tighter bounds on the suboptimality of the synthesized sensor selection policy.

\bibliographystyle{IEEEtran}
\bibliography{refs2}

\begin{thebibliography}{10}
\providecommand{\url}[1]{#1}
\csname url@samestyle\endcsname
\providecommand{\newblock}{\relax}
\providecommand{\bibinfo}[2]{#2}
\providecommand{\BIBentrySTDinterwordspacing}{\spaceskip=0pt\relax}
\providecommand{\BIBentryALTinterwordstretchfactor}{4}
\providecommand{\BIBentryALTinterwordspacing}{\spaceskip=\fontdimen2\font plus
\BIBentryALTinterwordstretchfactor\fontdimen3\font minus
  \fontdimen4\font\relax}
\providecommand{\BIBforeignlanguage}[2]{{%
\expandafter\ifx\csname l@#1\endcsname\relax
\typeout{** WARNING: IEEEtran.bst: No hyphenation pattern has been}%
\typeout{** loaded for the language `#1'. Using the pattern for}%
\typeout{** the default language instead.}%
\else
\language=\csname l@#1\endcsname
\fi
#2}}
\providecommand{\BIBdecl}{\relax}
\BIBdecl

\bibitem{freundlich2015optimal}
C.~Freundlich, P.~Mordohai, and M.~M. Zavlanos, ``Optimal path planning and
  resource allocation for active target localization,'' in \emph{2015 American
  Control Conference (ACC)}.\hskip 1em plus 0.5em minus 0.4em\relax IEEE, 2015,
  pp. 3088--3093.

\bibitem{atanasov2013nonmyopic}
N.~Atanasov, B.~Sankaran, J.~L. Ny, G.~J. Pappas, and K.~Daniilidis,
  ``Nonmyopic view planning for active object detection,'' \emph{arXiv preprint
  arXiv:1309.5401}, 2013.

\bibitem{kantaros2019asymptotically}
Y.~Kantaros, B.~Schlotfeldt, N.~Atanasov, and G.~J. Pappas, ``Asymptotically
  optimal planning for non-myopic multi-robot information gathering.'' in
  \emph{Robotics: Science and Systems}, 2019.

\bibitem{williamson2011design}
D.~P. Williamson and D.~B. Shmoys, \emph{The design of approximation
  algorithms}.\hskip 1em plus 0.5em minus 0.4em\relax Cambridge university
  press, 2011.

\bibitem{zhang2017sensor}
H.~Zhang, R.~Ayoub, and S.~Sundaram, ``Sensor selection for {K}alman filtering
  of linear dynamical systems: Complexity, limitations and greedy algorithms,''
  \emph{Automatica}, 2017.

\bibitem{joshi2008sensor}
S.~Joshi and S.~Boyd, ``Sensor selection via convex optimization,'' \emph{IEEE
  Transactions on Signal Processing}, 2008.

\bibitem{mo2011sensor}
Y.~Mo, R.~Ambrosino, and B.~Sinopoli, ``Sensor selection strategies for state
  estimation in energy constrained wireless sensor networks,''
  \emph{Automatica}, 2011.

\bibitem{roy2016spatio}
V.~Roy, A.~Simonetto, and G.~Leus, ``Spatio-temporal sensor management for
  environmental field estimation,'' \emph{Signal Processing}, 2016.

\bibitem{jawaid2015submodularity}
S.~T. Jawaid and S.~L. Smith, ``Submodularity and greedy algorithms in sensor
  scheduling for linear dynamical systems,'' \emph{Automatica}, 2015.

\bibitem{shamaiah2010greedy}
M.~Shamaiah, S.~Banerjee, and H.~Vikalo, ``Greedy sensor selection: Leveraging
  submodularity,'' in \emph{49th IEEE conference on decision and control
  (CDC)}.\hskip 1em plus 0.5em minus 0.4em\relax IEEE, 2010.

\bibitem{tzoumas2016sensor}
V.~Tzoumas, A.~Jadbabaie, and G.~J. Pappas, ``Sensor placement for optimal
  {K}alman filtering: Fundamental limits, submodularity, and algorithms,'' in
  \emph{2016 American Control Conference (ACC)}.\hskip 1em plus 0.5em minus
  0.4em\relax IEEE, 2016.

\bibitem{hashemi2020randomized}
A.~Hashemi, M.~Ghasemi, H.~Vikalo, and U.~Topcu, ``Randomized greedy sensor
  selection: Leveraging weak submodularity,'' \emph{IEEE Transactions on
  Automatic Control}, 2020.

\bibitem{chamon2020approximate}
L.~F. Chamon, G.~J. Pappas, and A.~Ribeiro, ``Approximate supermodularity of
  {K}alman filter sensor selection,'' \emph{IEEE Transactions on Automatic
  Control}, 2020.

\bibitem{singh2017supermodular}
P.~Singh, M.~Chen, L.~Carlone, S.~Karaman, E.~Frazzoli, and D.~Hsu,
  ``Supermodular mean squared error minimization for sensor scheduling in
  optimal {K}alman filtering,'' in \emph{2017 American Control Conference
  (ACC)}.\hskip 1em plus 0.5em minus 0.4em\relax IEEE, 2017.

\bibitem{vitus2012efficient}
M.~P. Vitus, W.~Zhang, A.~Abate, J.~Hu, and C.~J. Tomlin, ``On efficient sensor
  scheduling for linear dynamical systems,'' \emph{Automatica}, 2012.

\bibitem{gupta2006stochastic}
V.~Gupta, T.~H. Chung, B.~Hassibi, and R.~M. Murray, ``On a stochastic sensor
  selection algorithm with applications in sensor scheduling and sensor
  coverage,'' \emph{Automatica}, vol.~42, no.~2, pp. 251--260, 2006.

\bibitem{pineau2003point}
J.~Pineau, G.~Gordon, S.~Thrun \emph{et~al.}, ``Point-based value iteration: An
  anytime algorithm for {P}{O}{M}{D}{P}s,'' in \emph{Proceedings of the
  Eighteenth International Joint Conference on Artificial Intelligence}, 2003.

\bibitem{vaisenberg2014scheduling}
R.~Vaisenberg, A.~Della~Motta, S.~Mehrotra, and D.~Ramanan, ``Scheduling
  sensors for monitoring sentient spaces using an approximate {P}{O}{M}{D}{P}
  policy,'' \emph{Pervasive and Mobile Computing}, 2014.

\bibitem{atia2011sensor}
G.~K. Atia, V.~V. Veeravalli, and J.~A. Fuemmeler, ``Sensor scheduling for
  energy-efficient target tracking in sensor networks,'' \emph{IEEE
  Transactions on Signal Processing}, 2011.

\bibitem{ghasemi2018perception}
M.~Ghasemi and U.~Topcu, ``Perception-aware point-based value iteration for
  partially observable {M}arkov decision processes,'' in \emph{Proceedings of
  the Twenty-Eighth International Joint Conference on Artificial Intelligence},
  2018.

\bibitem{bertsekas2011dynamic}
D.~P. Bertsekas, ``Dynamic programming and optimal control 3rd edition, volume
  ii,'' \emph{Belmont, MA: Athena Scientific}, 2011.

\bibitem{completeness_lecture}
\BIBentryALTinterwordspacing
J.-M. Liou, ``Space of bounded functions and space of continuous functions,''
  University Lecture. [Online]. Available:
  \url{http://www.math.ncku.edu.tw/~fjmliou/Complex/spbounded.pdf}
\BIBentrySTDinterwordspacing

\bibitem{khamsi2011introduction}
M.~A. Khamsi and W.~A. Kirk, \emph{An introduction to metric spaces and fixed
  point theory}.\hskip 1em plus 0.5em minus 0.4em\relax John Wiley \& Sons,
  2011.

\bibitem{andrien2020near}
A.~Andri{\"e}n and D.~J. Antunes, ``Near-optimal {MAP} estimation for {M}arkov
  jump linear systems using relaxed dynamic programming,'' \emph{IEEE Control
  Systems Letters}, vol.~4, no.~4, pp. 815--820, 2020.

\end{thebibliography}

\appendix

\subsection{Supplementary results}
\label{sec:a1}

\begin{lemma}
\label{lem:matrixinversion}
For $M\succ 0$, $N \succ 0$, $X\succeq 0$ and $A\in \mathbb{R}^{n\times n}$, the following inequality holds:
\begin{align*}
&((M+AXA^\top)^{-1}+N)^{-1} \\
&\preceq (M^{-1}+N)^{-1} \\
&\qquad +(M^{-1}+N)^{-1}M^{-1}AXA^\top M^{-1}(M^{-1}+N)^{-1}.
\end{align*}
Moreover, if $A$ is Schur stable, we have
\[
\text{Tr}((M^{-1}+N)^{-1}M^{-1}AXA^\top M^{-1}(M^{-1}+N)^{-1}) \leq \text{Tr}(X).
\]

\end{lemma}

\begin{proof}
First, recall the general matrix inversion lemma:
\begin{equation}\label{eq:gen_mat_inv_lem}
    (A + BCB^{\top})^{-1} = A^{-1} - A^{-1}B(B^{\top}A^{-1}B + C^{-1})^{-1}B^{\top}A^{-1}
\end{equation}
Now, set $F:=X^{\frac{1}{2}}A^\top$. Then one can show that
\begin{align}
&((M+AXA^\top)^{-1}+N)^{-1} \label{lemma3-1} \\
&=((M+F^\top F)^{-1}+N)^{-1} \label{lemma3-2} \\
&=(M^{-1}+N-M^{-1}F^\top (FM^{-1}F^\top+I)^{-1}FM^{-1})^{-1} \label{lemma3-3} \\
&=(M^{-1}+N)^{-1} \nonumber \\
&\quad+(M^{-1}+N)^{-1}M^{-1}F^\top \nonumber \\
& \qquad \times (I+FM^{-1}(M-(M^{-1}+N)^{-1})M^{-1}F^\top)^{-1} \nonumber \\
& \qquad \times FM^{-1}(M^{-1}+N)^{-1} \label{lemma3-4} \\
&\preceq  (M^{-1}+N)^{-1} \nonumber \\
& \quad +(M^{-1}+N)^{-1}M^{-1}F^\top FM^{-1}(M^{-1}+N)^{-1}  \label{lemma3-5} \\
&= (M^{-1}+N)^{-1} \nonumber \\
& \quad +(M^{-1}+N)^{-1}M^{-1}AXA^\top M^{-1}(M^{-1}+N)^{-1} \label{lemma3-6}
\end{align}
where \eqref{lemma3-2} follows \eqref{lemma3-1} by simply substituting our relation for $F$. We then obtain \eqref{lemma3-3} by applying the general matrix inversion lemma in \eqref{eq:gen_mat_inv_lem} to \eqref{lemma3-2} wherein we use $A=M$, $B=F$, and $C=I$ (where $I$ is the identity matrix). To obtain \eqref{lemma3-4}, we again apply the general form of the matrix inversion lemma, this time setting $A=(M^{-1}+N)$, $B=FM^{-1}$, and $C=-(FM^{-1}F^{\top} + I)^{-1}$, and subsequently rearranging terms. The following inequality in \eqref{lemma3-5} is obtained by noting that $M-(M^{-1}+N)^{-1}\succeq 0$, which implies $(I+FM^{-1}(M-(M^{-1}+N)^{-1})M^{-1}F^\top)^{-1} \preceq I^{-1} = I$. Substituting $I$ yields the result. Finally, we obtain \eqref{lemma3-6} by substituting back in our relation for $F$.
To show the second claim, first notice that all the eigenvalues of the matrix $Z:=M^{\frac{1}{2}}(I+M^{\frac{1}{2}}NM^{\frac{1}{2}})^{-1}M^{-\frac{1}{2}}$ satisfy $0< \lambda_i(Z)<1$. Thus,
\begin{align*}
0&\prec M^{-1}(M^{-1}+N)^{-2}M^{-1} \\
&=M^{-\frac{1}{2}}(I+M^{\frac{1}{2}}NM^{\frac{1}{2}})^{-1}M^{\frac{1}{2}}M^{\frac{1}{2}}(I+M^{\frac{1}{2}}NM^{\frac{1}{2}})^{-1}M^{-\frac{1}{2}} \\
&=Z^\top Z \prec I.
\end{align*}
Therefore,
\begin{align*}
&\text{Tr}((M^{-1}+N)^{-1}M^{-1}AXA^\top M^{-1}(M^{-1}+N)^{-1}) \\
&\leq \text{Tr}(M^{-1}(M^{-1}+N)^{-2}M^{-1}AXA^\top) \\
&\leq \text{Tr}(AXA^\top) \\
&\leq \text{Tr}(X),
\end{align*}
completing the proof.
\end{proof}

\begin{lemma}
\label{lem:f_ub}
Assume $A\in\mathbb{R}^{n \times n}$ is Schur stable.
For every $P\succeq 0$, $\Delta P \succeq 0$ and $S$, we have
$f(P+\Delta P)\preceq f(P,S)+\Delta f$ where $\Delta f \succeq 0$ and $\text{Tr}(\Delta f)\leq \text{Tr}(\Delta P)$.
\end{lemma}
\begin{proof}
This result follows from Lemma~\ref{lem:matrixinversion} by setting $M=APA^\top+W$, $N=C_S^\top V_S^{-1}C_S$, $X=\Delta P$ and $\Delta f=(M^{-1}+N)^{-1}M^{-1}AXA^\top M^{-1}(M^{-1}+N)^{-1}$.
\end{proof}

\end{document}